\documentclass[a4paper]{article}
\usepackage[margin=2cm]{geometry}
\usepackage[utf8]{inputenc}
\usepackage[english]{babel} % English language/hyphenation
\usepackage{amsmath,amssymb,amsthm} % Math packages
\usepackage{esint}
\usepackage[colorlinks]{hyperref}
\usepackage[color=yellow]{todonotes}
\usepackage{enumerate}
\usepackage{framed}
\usepackage{upgreek}
\usepackage{pgfplots}
\usepackage{float}
\usepackage{tikz,tikz-cd}

\RequirePackage[authoryear]{natbib}

\newtheorem{theorem}{Theorem}

\newtheorem{definition}{Definition}[section]

\newtheorem{lemma}[theorem]{Lemma}

\newtheorem{remark}[theorem]{Remark}

\usepackage{parskip}
\pgfplotsset{compat=1.14}

\begin{document}
\title{Stochastic geometric mechanics with diffeomorphisms\\ \bigskip\Large
Darryl D. Holm and Erwin Luesink 
\\ \bigskip\small
Department of Mathematics, Imperial College London SW7 2AZ, UK\\
email: d.holm@ic.ac.uk,\ e.luesink16@imperial.ac.uk
}
\date{}                                           % Activate to display a given date or no date

\maketitle
%\section{}
%\subsection{}

\makeatother

\begin{abstract}
Noether's celebrated theorem associating symmetry and conservation laws in classical field theory is adapted to allow for broken symmetry in geometric mechanics and is shown to play a central role in deriving and understanding the generation of fluid circulation via the Kelvin-Noether theorem for ideal fluids with stochastic advection by Lie transport (SALT).
\end{abstract}

\section{Noether's theorem in geometric mechanics}
\subsection{Euler-Poincar\'e reduction}
Geometric mechanics deals with group-invariant variational principles. In this setting, Noether's theorem \cite{noether1918invariante, noether1971invariant} plays a key role. Given the tangent lift action $G\times TM\to TM$ of a Lie group $G$ on the tangent bundle $TM$ of a manifold $M$\footnote{$M$ is called the \emph{configuration manifold} in classical mechanics.} on which $G$ acts transitively, Noether's theorem states that each Lie symmetry of a Lagrangian $L:TM\to \mathbb{R}$ defined  in the action integral $S=\int L(q,v)dt$ for Hamilton's variational principle $\delta S = 0$ with $(q,v)\in TM$ implies a conserved quantity for the corresponding Euler-Lagrange equations defined on the cotangent bundle $T^*M$.  
%Quantities which are conserved by the Euler-Lagrange equations which emerge from Hamilton's variational principle can be used the reduce the number of degrees of freedom in a mechanical problem. When the configuration space of a mechanical problem is a Lie group and the action functional is invariant under the action of the Lie group, the number of degrees of freedom can be halved by a procedure that is now called Euler-Poincar\'e reduction. 
The conserved quantities arising from Noether's theorem in the case where the configuration manifold $M$ is a Lie group $G$ were studied by Smale, in \cite{smale1970topologya, smale1970topologyb}, where it was shown that the reduction procedure $TG\to TG\setminus G\simeq\mathfrak{g}$ leads to dynamics which take place on the dual $\mathfrak{g}^*$ of the Lie algebra $\mathfrak{g}$. The dynamical variable $m\in\mathfrak{g}^*$ in the dual Lie algebra is now called the momentum map (Smale called it angular momentum). In general, the configuration manifold $M$ is not a Lie group. However, when a Lie group $G$ acts transitively on a configuration manifold $M$ the proof of Noether's theorem induces a cotangent-lift momentum map $J: T^*M\to\mathfrak{g}^*$. The momentum map induced this way is an infinitesimally equivariant Poisson map taking functions on the cotangent bundle $T^*M$ of $M$ to the dual Lie algebra $\mathfrak{g}^*$ of the Lie group $G$. The momentum map $J: T^*M\to\mathfrak{g}^*$ is equivariant and Poisson, even if $G$ is not a Lie symmetry of the Lagrangian in Hamilton's principle.  Momentum maps naturally lead from the Lagrangian side to the Hamiltonian side. The Hamiltonian dynamics on $T^*M$ involves symplectic transformations.  However, as we shall discuss below, for the class of Hamiltonians which can be defined as $H\circ J: \mathfrak{g}^*\to \mathbb{R}$, the momentum map induces Euler-Poincar\'e motion on the Lagrangian side and Lie-Poisson motion on the Hamiltonian side. To illustrate these remarks, we return to the situation in which the configuration manifold, $M$, is a Lie group, $G$.

For hyperregular Lagrangians, the Legendre transform to the Hamiltonian side is invertible and one may reconstruct the solution on $G$ from its representation on $T^*G\setminus G\simeq\mathfrak{g}^*$. In that case, solving the equations describing the evolution of the momentum map on the dual Lie algebra $\mathfrak{g}^*$ is equivalent to solving the equations on the  cotangent bundle $T^*G$ when the configuration manifold is $G$. When the Lie group $G$ acts transitively, freely and properly on the configuration manifold $M$, then one may reconstruct the solution on $M$ from its representation on $T^*G\setminus G\simeq\mathfrak{g}^*$. The last statement is proven for finite-dimensional Lie groups $G$ in, e.g., \cite{abraham1978foundations}. 

The Lie-group reduced equations defined on the dual Lie algebra $\mathfrak{g}^*$ via Smale's  procedure of reduction by symmetry $T^*G\setminus G\simeq\mathfrak{g}^*$ are called Euler-Poincar\'e equations after \cite{poincare1901forme}. Provided the Lagrangian is hyperregular, the Euler-Poincar\'e reduction procedure can be expressed in terms of the cube of linked commutative diagrams shown in figure \ref{fig:cube}. 
\begin{figure}[H]
\small
\centering
\begin{tikzcd}[row sep=3em, column sep=small]
& 
L:TG\to\mathbb{R} \arrow[dl] \arrow[rr,  "\text{Legendre transform}", leftrightarrow] \arrow[dd] 
& 
& 
H:T^*G\to\mathbb{R} \arrow[dl] \arrow[dd]
\\
\text{Euler-Lagrange eqns} \arrow[rr, crossing over, Leftrightarrow] 
& 
& \text{Hamilton's eqns}
\\
& \ell:\mathfrak{g}\to\mathbb{R} \arrow[dl] \arrow[rr, "\text{Legendre \hspace{0.25cm} transform}", leftrightarrow] 
& 
& \hslash:\mathfrak{g}^*\to\mathbb{R} \arrow[dl]
\\
\text{Euler-Poincar\'e eqns} \arrow[rr, Leftrightarrow] \arrow[from=uu, crossing over]
& 
& \text{Lie-Poisson eqns} \arrow[from=uu, crossing over]
\end{tikzcd}
\caption{The cube of commutative diagrams for geometric mechanics on Lie groups. Euler-Poincar\'e reduction (on the left side) and Lie-Poisson reduction (on the right side) are both indicated by the arrows pointing down. The diagrams are all commutative, provided the Legendre transformation and reduced Legendre transformation are both invertible.}
\label{fig:cube}
\end{figure}
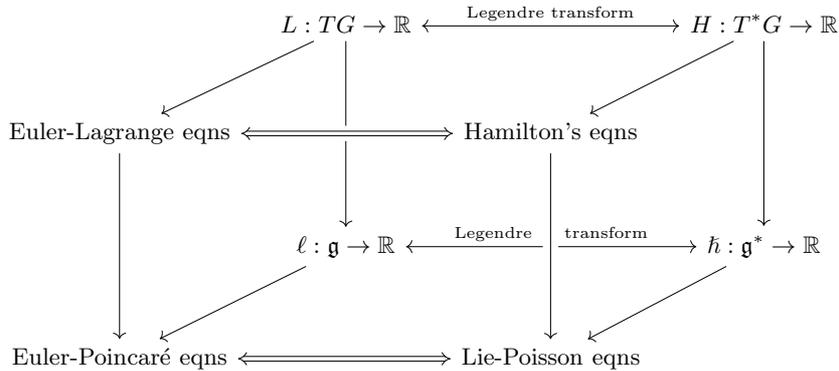
To summarise the notation in figure \ref{fig:cube}, $G$ denotes the  configuration manifold which is assumed to be isomorphic to a Lie group, $TG$ is the tangent bundle, $T^*G$ is the cotangent bundle, $TG\setminus G \simeq \mathfrak{g}$ is the Lie algebra and $T^*G\setminus G\simeq \mathfrak{g}^*$ is the dual of the Lie algebra. The Lagrangian is a functional $L:TG\to\mathbb{R}$ and the Hamiltonian is a functional $H:T^*G\to\mathbb{R}$. Euler-Poincar\'e reduction takes advantage of Lie group symmetries to transform the Lagrangian and Hamiltonian into group-invariant variables, which leads to a reduced Lagrangian $\ell:\mathfrak{g}\to\mathbb{R}$ and a reduced Hamiltonian $\hslash:\mathfrak{g}^*\to\mathbb{R}$. The diagram comprising the face of the cube involving these functionals in figure \ref{fig:cube} commutes if the Legendre transform is a diffeomorphism. This is guaranteed if the Lagrangian or Hamiltonian is hyperregular. The Euler-Lagrange equations and Hamilton's equations are related via a change of variables, which also holds for the Euler-Poincar\'e equations and the Lie-Poisson equations. Many finite dimensional mechanical systems may be described naturally in this framework. The classic example is the rotating rigid body, discussed from the viewpoint of symmetry reduction by Poincar\'e in \cite{poincare1901forme}. In his 1901 paper, Poincar\'e also raised the issue of \emph{symmetry breaking}, by introducing the vertical acceleration of gravity, which breaks the  $SO(3)$  symmetry for free rotation and restricts it to  $SO(2)$ for rotations about the vertical axis.

Stochasticity may also be included in the framework of Euler-Poincar\'e reduction by symmetry. The first attempt to include noise consistently in finite-dimensional symplectic Hamiltonian  mechanics was by \cite{bismut1982mecanique} and reduction by symmetry of stochastic systems was studied by \cite{lazaro2008stochastic}.  

\paragraph{Plan of the paper.}
In the present work, we will review Euler-Poincar\'e reduction of stochastic infinite dimensional variational systems with symmetry breaking.  The infinite dimensional case is interesting because it is the natural setting for fluid dynamics, quantum mechanics and elasticity. The foundations of the finite dimensional stochastic geometric mechanics are established in \citet{cruzeiro2018momentum}. We will explore the infinite dimensional case in context of fluid dynamics, where symmetry under the smooth invertible maps of the flow domain is broken by the spatial dependence of the initial mass density. 

\subsection{Sobolev class diffeomorphisms}
Consider an $n$-dimensional compact and oriented smooth manifold $M$, equipped with a Riemannian metric $\langle\,\cdot\,,\,\cdot\,\rangle$. This will be the spatial domain of flow and $X\in M$ will denote the initial position of any given fluid particle. The manifold $M$ is acted upon by a group of Sobolev class diffeomorphisms. In \cite{ebin1970groups} it is shown that the space of $C^\infty$ diffeomorphisms, defined by $\mathfrak{D}=\{g\in C^\infty(M,M)|\,g\text{ is bijective and } g^{-1}\in C^\infty(M,M)\}$, is not the convenient setting to study fluid dynamics, but that one should use $\mathfrak{D}^s = \{g\in H^s(M,M)|\,g\text{ is bijective and } g^{-1}\in H^s(M,M)\}$, the space of Sobolev class diffeomorphisms with $s$ weak derivatives. The reason for this choice is that the smooth diffeomorphisms constitute a Fr\'echet manifold for which there is no inverse or implicit function theorem and no general solution theorem for ordinary differential equations. Each of these latter features would prohibit the study of geodesics. 

The space of Sobolev class diffeomorphisms is both a Hilbert manifold and a topological group if $s>n/2+1$, as was shown by \cite{ebin1967space}. The Hilbert manifold structure implies the existence of function inverses and the implicit function theorem, as well as the existence of a general solution theorem for ordinary differential equations. This additional structure also implies that one can construct the tangent space of $\mathfrak{D}^s$ in the usual way and study geodesics. The space $\mathfrak{D}^s$ is the configuration space for continuum mechanics and each $g\in\mathfrak{D}^s$ is called a configuration. A fluid trajectory starting from $X\in M$ at time $t=0$ is given by $x(t)=g_t(X)=g(X,t)$, with $\mathfrak{D}^s\ni g:M\times\mathbb{R}^+\to M$ being a continuous one-parameter subgroup of $\mathfrak{D}^s$. In the deterministic case, computing the time derivative of this one-parameter subgroup gives rise to the \emph{reconstruction equation}, given by
\begin{equation}
\frac{\partial}{\partial t}g_t(X) = u(g_t(X),t),
\label{eq:reconstructiondeterministic}
\end{equation}
where $u_t(\,\cdot\,)=u(\,\cdot\,,t)\in \mathfrak{X}^s$ is a time dependent vector field with flow $g_t(\,\cdot\,)=g(\,\cdot\,,t)$. The initial data is given by $g(X,0)=X$. Here $\mathfrak{X}^s=H^s(TM)$ denotes the space of Sobolev class vector fields on $M$, which is also the Lie algebra associated to the Sobolev class diffeomorphisms. 

\subsection{Stochastic advection by Lie transport (SALT)}
In the setting of stochastic advection by Lie transport (SALT), which was introduced by \cite{holm2015variational}, the deterministic reconstruction equation in \eqref{eq:reconstructiondeterministic} is replaced by the semimartingale
\begin{equation}
{\sf d}g(X,t) = u(g_t(X),t)dt + \sum_{i=1}^M \xi_i(g_t(X))\circ dW_t^i,
\label{eq:reconstructionstochastic}
\end{equation}
where the symbol $\circ$ means that the stochastic integral is taken in the Stratonovich sense. The initial data is given by $g(X,0)=X$. The $W_t^i$ are independent, identically distributed Brownian motions, defined with respect to the standard stochastic basis $(\Omega,\mathcal{F},(\mathcal{F}_t)_{t\geq 0},\mathbb{P})$. Such a noise was shown to arise from a multi-time homogenisation argument in \cite{cotter2017stochastic}. The $\xi_i(\,\cdot\,)\in\mathfrak{X}^s$ are called data vector fields and are prescribed. These data vector fields represent the effects of unresolved degrees of freedom on the resolved scales of fluid motion and account for unrepresented processes. They are determined by applying empirical orthogonal function analysis to appropriate numerical and/or observational data. For instance, for an application to the two dimensional Euler equations for an ideal fluid, see \cite{cotter2019numerically} and for an application to a two-layer quasi-geostrophic model, see \cite{cotter2018modelling}. Stochastic models enable the use of a variety of methods in data assimilation, which are discussed in \cite{cotter2019particle}. It is not difficult to make sense of \eqref{eq:reconstructiondeterministic}, but understanding \eqref{eq:reconstructionstochastic} is more complicated. In \cite{de2020implications}, a \emph{stochastic chain rule} is shown to exist. This stochastic chain rule is called the \emph{Kunita-It\^o-Wentzell (KIW) formula} and helps interpret the semimartingale in \eqref{eq:reconstructionstochastic}. The KIW formula will also be used later to prove the stochastic Kelvin circulation theorem. First, however, the space $\mathfrak{D}^s$ needs to be given more structure.

The space $\mathfrak{D}^s$ inherits a \emph{weak Riemannian structure} from the underlying manifold $M$ in a natural way. For $g\in \mathfrak{D}^s$ and $V,W\in T_g\mathfrak{D}^s$, one can define the following bilinear form
\begin{equation}
(V,W)=\int_M \langle(V(X),W(X)\rangle_{g(X)}\mu(dX),
\label{eq:weakriemannian}
\end{equation}
where $\mu$ is the volume form on $M$ induced by the metric. The Riemannian structure induced by \eqref{eq:weakriemannian} is weak because the topology is of type $L^2$, which is strictly weaker than the $H^s$ topology. This bilinear form is a linear functional on the Hilbert space $T_g\mathfrak{D}^s$ and can be used to define the dual space $T_g^*\mathfrak{D}^s$. The pairing between $V\in T_g\mathfrak{D}^s$ and $\alpha\in T_g^*\mathfrak{D}^s$ is given by 
\begin{equation}
\langle \alpha,V \rangle = \int_M \alpha(X)\cdot V(X).
\label{eq:dualitypairing}
\end{equation}
Hence the metric on $M$ and the volume form $\mu(dX)$ can be used to construct the isomorphism between $T\mathfrak{D}^s$ and $T^*\mathfrak{D}^s$ as $V(X)\mapsto\alpha(X)=V^\flat(X)\mu(dX)$, where $\flat:TM\to T^*M$ is one of the musical isomorphisms that are induced by the metric on $M$. The group $\mathfrak{D}^s$ is not a Lie group; since right multiplication is smooth, but left multiplication is only continuous. Hence $\mathfrak{D}^s$ is a topological group with a weak Riemannian structure. In general, these properties are not sufficient to guarantee the existence of an exponential map. However, \cite{ebin1970groups} showed that an exponential map can exist in many important cases. In particular, they showed that the geodesic spray associated to \eqref{eq:weakriemannian} (with and without forcing) is smooth.\footnote{The \emph{geodesic spray} is the vector field whose integral curves are the geodesics.} The smoothness of the geodesic spray persists even though $H^s$ diffeomorphisms are considered rather than smooth diffeomorphisms. Combined with the existence of an exponential map, the smoothness property implies a regular interpretation of the Euler-Poincar\'e equations on $\mathfrak{D}^s$, provided that one uses right translations and right representations of the group on itself and its Lie algebra, as shown in \cite{holm1998euler}. However, due to the presence of the volume form $\mu(dm)$, the bilinear form \eqref{eq:weakriemannian} is not right-invariant under the action of the entire $H^s$ diffeomorphism group, although there is right-invariance under the action of the isotropy subgroup $\mathfrak{D}^s_\mu = \{g\in\mathfrak{D}^s | \, g_*\mu = \mu\}$. Since this subgroup is a proper subgroup, as it is smaller than $\mathfrak{D}^s$ itself. Thus, one speaks of \emph{symmetry breaking}. 

In deriving the equations of ideal deterministic fluid dynamics, one needs to keep track of the volume form as well. The appropriate mathematical setting for this is an \emph{outer semidirect product group}. This means that one constructs a new group from two given groups with a particular type of group operation. For continuum mechanics, the ingredients are $\mathfrak{D}^s$ and $V^*$, where $V^*$ is a vector space of tensor fields. This vector space is the space of \emph{advected quantities} and it will always contain at least the volume form $\mu$.  \smallskip

\begin{definition}[Advected quantity]
A fluid variable is said to be \emph{advected}, if it keeps its value along Lagrangian particle trajectories. Advected quantities are sometimes called \emph{tracers}, because the evolution histories of scalar advected quantities with different initial values (labels) trace out the Lagrangian particle trajectories of each label, or initial value, via the \emph{push-forward} of the full diffeomorphism group, i.e., $a_t=g_{t\,*}a_0= a_0g_t^{-1}$, where $g_t$ is a time-dependent curve on the manifold of diffeomorphisms that represents the fluid flow.
\end{definition}\smallskip

\begin{remark}[Advected quantities as order parameters]
When several advected quantities are involved, the space $V^*$ is the direct sum of several vector spaces, where each summand space hosts a different advected quantity. In general, each additional advected quantity decreases the dimension of the isotropy subgroup. For example, consider an ideal deterministic fluid with a buoyancy variable $b$, then the Lagrangian corresponding to the model will depend on $\mu$ and $b$ in a parametric manner. This Lagrangian will be right invariant under the action of the isotropy subgroup $\mathfrak{D}^s_{\mu,b} = \{g\in\mathfrak{D}^s|\, g_*\mu=\mu \text{ and } g_*b=b\}$. Hence, advected quantities are \emph{order parameters} and each additional order parameter breaks more symmetry. For the sake of notation, one usually writes $\mathfrak{D}^s_{a_0}$ for the isotropy subgroup, no matter how many advected quantities there are. One then uses $a$ to represent all advected quantities and $a_0$ to denote the initial value of the advected quantities. 
\end{remark}

\subsection{Semidirect product group adjoint \& coadjoint actions}
The semidirect product group action is constructed in the following way. The representation of $\mathfrak{D}^s$ on a vector space $V$ is by push-forward, which is a left representation, as shown by \cite{marsden1984semidirect}. The representation of the group on itself and on its Lie algebra is a right representation. In terms of analysis, this means that all representations are smooth and no derivatives need to be counted. The group action of the semidirect product group is given by
\begin{equation}
\begin{aligned}
\bullet:(\mathfrak{D}^s\times V)\times(\mathfrak{D}^s\times V)\to(\mathfrak{D}^s\times V)\\
\quad (g_1,v_1)\bullet(g_2,v_2):= (g_1\circ g_2,v_2+(g_2)_*v_1)
\end{aligned}
\label{eq:semidirectproduct}
\end{equation}
with $g_1,g_2\in\mathfrak{D}^s$ and $v_1,v_2\in V$. The semidirect product group is often denoted as $\mathfrak{D}^s\circledS V = (\mathfrak{D}^s\times V,\bullet)$. In the group action above, $(g_2)_*v_1$ denotes the \emph{push-forward} of $v_1$ by $g_2$ and $\circ$ denotes composition. Note that the group affects both slots in \eqref{eq:semidirectproduct}, but the vector space only appears in the second slot. The identity element of the semidirect product group is $(e,0)$ where $e\in\mathfrak{D}^s$ is the identity diffeomorphism and $0\in V$ is the zero vector. An inverse element is given by
\begin{equation}
(g,v)^{-1} = (g^{-1},-(g^{-1})_*v) = (g^{-1}, -g^*v),
\end{equation}
where $g^*v$ denotes the pull-back of $v$ by $g$. To understand how reduction works for semidirect products, it is helpful to know how the group acts on its Lie algebra and on the dual of its Lie algebra. Duality will be defined with respect to the sum of the pairing \eqref{eq:dualitypairing} and the dual linear transformation $[\,\cdot\,]^*$ on $V$. This pairing induces another pairing in a natural way on $\mathfrak{X}^s\times V$. Consider two at least $C^1$ one parameter subgroups $(g_t,v_t),(\widetilde{g}_\epsilon,\widetilde{v}_\epsilon)\in \mathfrak{D}^s\times V$. Using these one parameter subgroups, one can compute the inner automorphism, or adjoint action of the group on itself. This adjoint action is defined by conjugation
\begin{equation}
\begin{aligned}
{\rm AD}:(\mathfrak{D}^s\times V)\times(\mathfrak{D}^s &\times V)\to(\mathfrak{D}^s\times V),\\
{\rm AD}_{(g_t,v_t)}(\widetilde{g}_\epsilon,\widetilde{v}_\epsilon) &:= (g_t,v_t)\bullet (\widetilde{g}_\epsilon,\widetilde{v}_\epsilon)\bullet (g_t,v_t)^{-1}\\
&= \big(g_t\circ\widetilde{g}_\epsilon\circ g_t^{-1}, g_t^*(\widetilde{v}_\epsilon - v_t + \widetilde{g}_{\epsilon*}v_t)\big).
\end{aligned}
\label{eq:AD}
\end{equation}
To see how the group acts on its Lie algebra, one can compute the derivative with respect to $\epsilon$ and evaluate at $\epsilon=0$ in the adjoint action of the group on itself. Let $\mathfrak{X}^s\ni \widetilde{u}=\frac{d}{d\epsilon}|_{\epsilon=0}\widetilde{g}_\epsilon$ and $V\ni\widetilde{b}=\frac{d}{d\epsilon}|_{\epsilon=0}\widetilde{v}_\epsilon$. This choice for a vector field is guided by the deterministic reconstruction equation in \eqref{eq:reconstructiondeterministic}. For any tensor $S_\epsilon\in T_s^r(M)$ whose dependence on $\epsilon$ is at least $C^1$ it holds that
\begin{equation}
\frac{d}{d\epsilon}\widetilde{g}_{\epsilon*}S_\epsilon = \widetilde{g}_{\epsilon*}\left(\frac{d}{d\epsilon}S_\epsilon-\mathcal{L}_{\widetilde{u}} S_\epsilon\right).
\label{eq:liechainrule}
\end{equation}
Important here is that the Lie derivative does not commute with pull-backs and push-forwards that depend on parameters, see \cite{abraham1978foundations}. The adjoint action of the group on its Lie algebra can be computed as 
\begin{equation}
\begin{aligned}
{\rm Ad}:(\mathfrak{D}^s\times V)\times(\mathfrak{X}^s &\times V)\to (\mathfrak{X}^s\times V),\\
{\rm Ad}_{(g_t,v_t)}(\widetilde{u},\widetilde{b})&:= \frac{d}{d\epsilon}\Big|_{\epsilon=0}{\rm AD}_{(g_t,v_t)}(\widetilde{g}_\epsilon,\widetilde{v}_\epsilon)\\
&= (g_{t*}\widetilde{u},g^*_t\widetilde{b}-g_t^*\mathcal{L}_{\widetilde{u}}v_t).
\end{aligned}
\label{eq:Ad}
\end{equation}
By means of the pairing on $\mathfrak{X}^s\times V$, one can compute the dual action to the adjoint action \eqref{eq:Ad}. This is called the coadjoint action of the group on the dual of its Lie algebra. Let $(\widetilde{m},\widetilde{a})\in(\mathfrak{X}^s\times V)^*$, then the coadjoint action is given by
\begin{equation}
\begin{aligned}
{\rm Ad}^*:(\mathfrak{D}^s\times V)\times(\mathfrak{X}^s &\times V)^*\to(\mathfrak{X}^s\times V)^*,\\
\langle{\rm Ad}^*_{(g_t^{-1},-g_t^{-1}v_t)}(\widetilde{m},\widetilde{a}),(\widetilde{u},\widetilde{b})\rangle &:= \langle(\widetilde{m},\widetilde{a}),{\rm Ad}_{(g_t,v_t)}(\widetilde{u},\widetilde{b})\rangle,\\
{\rm Ad}^*_{(g_t^{-1},-g_t^{-1}v_t)}(\widetilde{m},\widetilde{a}) &= (g_t^*\widetilde{m}+v_t\diamond g_{t*}\widetilde{a},g_{t*}\widetilde{a}).
\end{aligned}
\label{eq:Ad*}
\end{equation}
\begin{definition}[The diamond operator]
The coadjoint action \eqref{eq:Ad*} features the \emph{diamond operator}, which is defined for $a\in V^*$, $u\in\mathfrak{X}^s$ and fixed $v\in V$ as
\begin{equation}
\langle v\diamond a, u\rangle_{\mathfrak{X}^{s*}\times\mathfrak{X}^s} := -\langle a,\mathcal{L}_u v\rangle_{V^*\times V}.
\end{equation}
Note that the diamond operator is the dual of the Lie derivative regarded as a map $\mathcal{L}_{(\,\cdot\,)}v:\mathfrak{X}^s\to V$, hence $v\diamond(\,\cdot\,):V^*\to\mathfrak{X}^{s*}$. The diamond operator shows how an element from the dual of the vector space acts on the dual of the Lie algebra. 
\end{definition}

When evaluated at $t=0$, the $t$-derivatives of ${\rm Ad}$ in \eqref{eq:Ad} and ${\rm Ad}^*$ in \eqref{eq:Ad*} define, respectively, the adjoint and coadjoint actions of the Lie algebra on itself and on its dual. Denote by $\mathfrak{X}^s\ni u = \frac{d}{dt}|_{t=0}g_t$ and $V\ni b=\frac{d}{dt}|_{t=0}v_t$. The adjoint action of the Lie algebra on itself is
\begin{equation}
\begin{aligned}
{\rm ad}:(\mathfrak{X}^s\times V)\times(\mathfrak{X}^s &\times V)\to (\mathfrak{X}^s\times V),\\
{\rm ad}_{(u,b)}(\widetilde{u},\widetilde{b})&:=\frac{d}{dt}\Big|_{t=0}{\rm Ad}_{(g_t,v_t)}(\widetilde{u},\widetilde{b}),\\
{\rm ad}_{(u,b)}(\widetilde{u},\widetilde{b})&=(-\mathcal{L}_u\widetilde{u},\mathcal{L}_u\widetilde{b}-\mathcal{L}_{\widetilde{u}}b)\\
&= (-[u,\widetilde{u}],\mathcal{L}_u\widetilde{b}-\mathcal{L}_{\widetilde{u}}b),
\end{aligned}
\label{eq:ad}
\end{equation}
where the bracket $[\,\cdot\,,\,\cdot\,]$ in \eqref{eq:ad} is the commutator of vector fields. The minus sign is due to fact that group acts on itself from the right. The coadjoint action of the Lie algebra on its dual can be obtained by computing the dual to \eqref{eq:ad} or by taking the derivative with respect to $t$ and evaluate at $t=0$ in \eqref{eq:Ad*}. Either way, one arrives at
\begin{equation}
\begin{aligned}
{\rm ad}^*:(\mathfrak{X}^s\times V)\times(\mathfrak{X}^s &\times V)^*\to(\mathfrak{X}^s\times V)^*,\\
\langle{\rm ad}^*_{(u,b)}(\widetilde{m},\widetilde{a}),(\widetilde{u},\widetilde{b})\rangle &:= \langle(\widetilde{m},\widetilde{a}),{\rm ad}_{(u,b)}(\widetilde{u},\widetilde{b})\rangle,\\
{\rm ad}^*_{(u,b)}(\widetilde{m},\widetilde{a}) &= (\mathcal{L}_u\widetilde{m} + b\diamond\widetilde{a},-\mathcal{L}_u\widetilde{a}),
\end{aligned}
\label{eq:ad*}
\end{equation}
in which \eqref{eq:ad} implies the last line in \eqref{eq:ad*}. Alternatively, one can also obtain \eqref{eq:ad*} by taking the derivative with respecto to $t$ in \eqref{eq:Ad*} and evaluate at $t=0$. \smallskip

\begin{remark}[Coadjoint action and the diamond operator]
The coadjoint action is an important operator in geometric mechanics and representation theory. It was shown by \cite{kirillov1962unitary} and in further work by \cite{kostant1970quantization} and \cite{souriau1970structure} that the coadjoint orbits of a Lie group $G$ have the structure of symplectic manifolds and are connected with Hamiltonian mechanics. See \cite{kirillov1999merits} for a review. The computations of the adjoint and coadjoint actions for the semidirect product group is valuable for fluid mechanics, as they introduce the two fundamental operators that appear in the equations of motion. The Lie derivative is responsible for transport of tensors along vector fields and its dual action given by the diamond operator encodes the symmetry breaking. In particular, the diamond operator introduces the effect of symmetry breaking into the Euler-Poincar\'e equations of motion.
\end{remark}

\section{Deterministic geometric fluid dynamics}
With the adjoint and coadjoint actions defined, one can derive continuum mechanics equations with advected quantities by using symmetry reduction. Euler-Poincar\'e reduction for a semidirect product group $\mathfrak{D}^s\times V$ as developed in \cite{holm1998euler} is sketched below in figure \ref{fig:cubesdp}. 
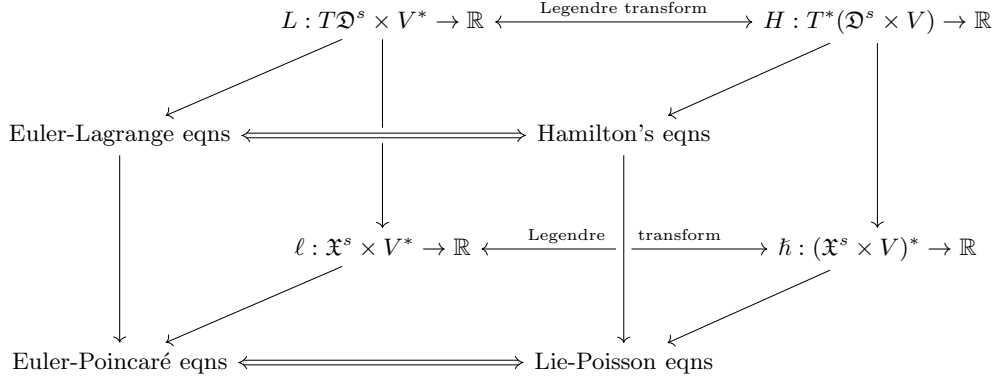
\begin{figure}[H]
\small
\centering
\begin{tikzcd}[row sep=3em, column sep=small]
& 
L:T\mathfrak{D}^s\times V^*\to\mathbb{R} \arrow[dl] \arrow[rr,  "\text{Legendre transform}", leftrightarrow] \arrow[dd] 
& 
& 
H:T^*(\mathfrak{D}^s\times V)\to\mathbb{R} \arrow[dl] \arrow[dd]
\\
\text{Euler-Lagrange eqns} \arrow[rr, crossing over, Leftrightarrow] 
& 
& \text{Hamilton's eqns}
\\
&
\ell:\mathfrak{X}^s\times V^*\to\mathbb{R} \arrow[dl] \arrow[rr, "\text{Legendre \hspace{0.25cm} transform}", leftrightarrow] 
& 
& 
\hslash:(\mathfrak{X}^s\times V)^*\to\mathbb{R} \arrow[dl] 
\\
\text{Euler-Poincar\'e eqns} \arrow[rr, Leftrightarrow] \arrow[from=uu, crossing over]
& 
& 
\text{Lie-Poisson eqns} \arrow[from=uu, crossing over]
\end{tikzcd}
\caption{The cube of continuum mechanics in the semidirect product group setting. Reduction is indicated by the arrows pointing down.}
\label{fig:cubesdp}
\end{figure}
 As shown by comparison of  figure \ref{fig:cubesdp} with figure \ref{fig:cube}, several new features arise in semidirect product Lie group reduction which differ from Euler-Poincar\'e reduction by symmetry when the configuration space itself is a Lie group. These differences can be conveniently explained by introducing the physical concept of an order parameter. As discussed earlier, the order parameters in continuum mechanics are the elements of $V^*$ which are advected by the action of the diffeomorphism group $\mathfrak{D}^s$. The advection is defined simply as the semidirect product action on the elements of $V^*$. The introduction of each additional advected state variable (or, order parameter) into the physical problem further breaks the original symmetry $\mathfrak{D}^s$. The remaining symmetry of the Lagrangian in Hamilton's principle is the isotropy subgroup $\mathfrak{D}^s_{a_0}$ of the initial conditions, $a_0$, for the entire set of advected quantities, $a$. The action of the diffeomorphism group $\mathfrak{D}^s$ on these initial conditions then describes their advection as the action of $\mathfrak{D}^s$ on its coset space $\mathfrak{D}^s\setminus\mathfrak{D}^s_{a_0}=V^*$. Once the inital values of the order parameters, $a_0$,  have been set, one must still define a Legendre transform from the Lagrangian formulation into the Hamiltonian formulation and vice versa. The Legendre transform in the setting of semidirect products is a partial Legendre transform, since it transforms between $T\mathfrak{D}^s$ and $T^*\mathfrak{D}^s$ or $T\mathfrak{D}^s\setminus\mathfrak{D}^s_{a_0} \simeq \mathfrak{X}^s$ and $T^*\mathfrak{D}^s\setminus\mathfrak{D}^s_{a_0} \simeq\mathfrak{X}^{s*}$ only after having fixed the value $a_0$ of the order parameters, which live in $V^*$. This coset reduction is what figure \ref{fig:cubesdp} shows. The remaining invariance of a functional under the action of the isotropy subgroup is called its \emph{particle relabelling symmetry}.

Our exploration continues on the Lagrangian side in figure \ref{fig:cubesdp}. Consider a fluid Lagrangian $L:T\mathfrak{D}^s\times V^*\to\mathbb{R}$. By fixing the value of $a_0\in V^*$, one can construct $L_{a_0}:T\mathfrak{D}^s\to\mathbb{R}$. If this Lagrangian is right invariant under the action of the isotropy subgroup $\mathfrak{D}_{a_0}^s$, then one can construct 
\begin{equation}
\begin{aligned}
L\left(\frac{d}{dt}g\circ g^{-1},e,a_0\right) &= L_{a_0}\left(\frac{d}{dt}g\circ g^{-1},e\right)\\
&= \ell_{a_0}\left(\frac{d}{dt}g\circ g^{-1}\right) = \ell\left(\frac{d}{dt}g\circ g^{-1}, g_*a_0\right).
\end{aligned}
\label{eq:lagrangians}
\end{equation}
Here $\circ$ means composition of functions. The same procedure applies to the Hamiltonian. Since the coadjoint action is known, it is straightforward to formulate the Lie-Poisson equations. The details of Hamiltonian semidirect product reduction and also more information on the Lagrangian semidirect product reduction can be found in \cite{holm1998euler}. 

The coadjoint action of the Lie algebra on its dual is also required for the Lagrangian semidirect product reduction. One can use the deterministic reconstruction equation to see that the argument of the Lagrangians in \eqref{eq:lagrangians} is
\begin{equation}
\frac{d}{dt}g\circ g^{-1} = u.
\end{equation}
Using this information, one can integrate the Lagrangian in time to construct the action functional. By requiring the variational derivative of the action functional to vanish, one can compute the equations of motion. However, due to the removal of symmetries, the variations are no longer free. 

\section{Stochastic geometric fluid dynamics}
In the situation where noise is present, that is, when the reconstruction equation is  \eqref{eq:reconstructionstochastic}, the Euler-Poincar\'e variations become stochastic. Consider $g:\mathbb{R}^2\to\mathfrak{D}^s$ with $g_{t,\epsilon}=g(t,\epsilon)$ to be a two parameter subgroup with smooth dependence on $\epsilon$, but only continuous dependence on $t$. Let us denote 
\[
{\sf d}\chi_{t,\epsilon}(X) = ({\sf d}g_{t,\epsilon}\circ g_{t,\epsilon})(X) = u_{t,\epsilon}(X)dt + \sum_{i=1}^N \xi_i(X)\circ dW_t^i
\]
 and  
 \[
 v_{t,\epsilon}(X) = (\frac{\partial}{\partial \epsilon}g_{t,\epsilon}\circ g_{t,\epsilon})(X)\,.
 \] 
 When a $\circ$ symbol is followed by $dW_t$ it means Stratonovich integration and in every other context the $\circ$ symbol is used to denote composition. Note that the data vector fields $\xi_i$ are prescribed and hence will not have a dependence on $\epsilon$. 
 
In order to compute with these stochastic subgroups and their associated vector fields, one needs a stochastic Lie chain rule. The Kunita-It\^o-Wentzell (KIW) formula is the stochastic generalisation of the Lie chain rule \eqref{eq:liechainrule}. A proof of the KIW formula is given  in \cite{de2020implications} for differential $k$-forms and vector fields. That proof includes the technical details on regularity that will be omitted here. In the KIW formula, the $k$-form is allowed to be a semimartingale itself. Let $K$ be a continuous adapted semimartingale that takes values in the $k$-forms and satisfies
\begin{equation}
K_t = K_0 + \int_0^t G_s ds + \sum_{i=1}^N\int_0^t H_{i\,s}\circ dB_s^i,
\label{eq:kformsemimartingale}
\end{equation}
where the $B_t^i$ are independent, identically distributed Brownian motions. The drift of the semimartingale $K$ is determined by $G$ and the diffusion by $H_i$, both of which are $k$-form valued continuous adapted semimartingales with suitable regularity. Let $g_t$ satisfy \eqref{eq:reconstructionstochastic}, then \cite{de2020implications} shows that the following holds
\begin{equation}
{\sf d}(g_t^*K_t) = g_t^*\big({\sf d}K_t + \mathcal{L}_{u_t} K_t\,dt + \mathcal{L}_{\xi_i}K_t \circ dW_t^i\big).
\label{eq:kiwformula}
\end{equation}
Equation \eqref{eq:kformsemimartingale} helps to interpret the ${\sf d}K_t$ term in the KIW formula \eqref{eq:kiwformula}. This formula will be particularly useful in computing the variations of the variables in the Lagrangian. To compute these variations, one needs the variational derivative. 

\paragraph{The variational derivative.} The variational derivative of a functional $F:\mathcal{B}\to\mathbb{R}$, where $\mathcal{B}$ is a Banach space, is denoted $\delta F/\delta \rho$ with $\rho\in\mathcal{B}$. The variational derivative can be defined by the first variation of the functional
\begin{equation}
\delta F[\rho]:= \frac{d}{d\epsilon}\Big|_{\epsilon=0} F[\rho+\epsilon \delta\rho] = \int \frac{\delta F}{\delta \rho}(x)\delta\rho(x)\,dx = \left\langle\frac{\delta F}{\delta \rho},\delta \rho\right\rangle.
\end{equation}
In the definition above, $\epsilon\in\mathbb{R}$ is a parameter, $\delta\rho\in\mathcal{B}$ is an arbitrary function and the first variation can be understood as a Fr\'echet derivative. A precise and rigorous definition can be found in \cite{gelfand2000calculus}. With the definition of the functional derivative in place, the following lemma can be formulated.
\medskip

\begin{lemma}
With the notation as above, the variations of $u$ and any advected quantity $a$ are given by 
\begin{equation}
\delta u(t) = {\sf d}v(t) + [{\sf d}\chi_t,v(t)],\quad \delta a(t) = -\mathcal{L}_{v(t)}a(t),
\label{def:delta-var}
\end{equation}
where $v(t)\in\mathfrak{X}^s$ is arbitrary.
\end{lemma}
\begin{proof}
The proof of the variation of $a(t)$ is a direct application of the Kunita-It\^o-Wentzell formula to $a(t,\epsilon)=g_{t,\epsilon*}a_0$. Note that the data vector fields $\xi_i$ are prescribed and do not depend on $\epsilon$. Denote by $x_{t,\epsilon} = g_{t,\epsilon}(X)$. Then one has
\begin{equation}
{\sf d}g_{t,\epsilon}(X) = {\sf d}x_{t,\epsilon} = u_{t,\epsilon}(x_{t,\epsilon})\,dt + \sum_{i=1}^N \xi_i(x_{t,\epsilon})\circ dW_t^i =: {\sf d}\chi_{t,\epsilon}(x_{t,\epsilon}).
\label{eq:twoparameterstochu}
\end{equation}
The vector field associated to the $\epsilon$-dependence of the two parameter subgroup is given by
\begin{equation}
\frac{\partial}{\partial \epsilon}g_{t,\epsilon} = \frac{\partial}{\partial \epsilon}x_{t,\epsilon} = v_{t,\epsilon}(x_{t,\epsilon}).
\label{eq:twoparameterstochv}
\end{equation}
Computing the derivative with respect to $\epsilon$ of \eqref{eq:twoparameterstochu} gives
\begin{equation}
\begin{aligned}
\frac{\partial}{\partial \epsilon}{\sf d}x_{t,\epsilon} &= \frac{\partial}{\partial \epsilon}\big({\sf d}\chi_{t,\epsilon}(x_{t,\epsilon})\big)\\
&= \left(\frac{\partial}{\partial \epsilon}u_{t,\epsilon} + v_{t,\epsilon}\cdot\frac{\partial}{\partial x_{t,\epsilon}}{\sf d}\chi_{t,\epsilon}\right)(x_{t,\epsilon}),
\end{aligned}
\end{equation}
where the independence of the data vector fields $\xi_i$ on $\epsilon$ was used. Taking the differential with respect to time of \eqref{eq:twoparameterstochv} gives
\begin{equation}
\begin{aligned}
{\sf d}\left(\frac{\partial}{\partial \epsilon} x_{t,\epsilon}\right) &= {\sf d}\big(v_{t,\epsilon}(x_{t,\epsilon})\big)\\
&=  \left( {\sf d}v_{t,\epsilon}(x_{t,\epsilon}) + {\sf d}\chi_{t,\epsilon}\cdot\frac{\partial}{\partial x_{t,\epsilon}}v_{t,\epsilon}\right)(x_{t,\epsilon}).
\end{aligned}
\end{equation}
One can then evaluate at $\epsilon=0$ and call upon equality of cross derivative-differential to obtain the result by subtracting. Since $g_{t,\epsilon}$ depends on $t$ in a $C^0$ manner, the integral representation is required. The particle relabelling symmetry permits one to stop writing the explicit dependence on space,
\begin{equation}
\delta u(t)\,dt = {\sf d}v(t) + [{\sf d}\chi_t,v(t)].
\end{equation}
This completes the proof of formula \eqref{def:delta-var} for the variation of $u(t)$.
\end{proof}
The notation in \eqref{eq:twoparameterstochu} needs careful explanation, because it comprises both a stochastic differential equation and a definition. The symbol ${\sf d}\chi_{t,\epsilon}$ is used to define a vector field, whereas ${\sf d}x_{t,\epsilon}$ denotes a stochastic differential equation. This lemma makes the presentation of the stochastic Euler-Poincar\'e theorem particularly simple. 
\medskip

\begin{theorem}[Stochastic Euler-Poincar\'e]\label{thm:SEP}
With the notation as above, the following are equivalent.
\begin{enumerate}[i)]
\item The constrained variational principle
\begin{equation}
\delta\int_{t_1}^{t_2}\ell(u,a)\,dt = 0
\end{equation}
holds on $\mathfrak{X}^s\times V^*$, using variations $\delta u$ and $\delta a$ of the form
\begin{equation}
\delta u = {\sf d}v + [{\sf d}\chi_t,v], \qquad \delta a = -\mathcal{L}_v a,
\end{equation}
where $v(t)\in \mathfrak{X}^s$ is arbitrary and vanishes at the endpoints in time for arbitrary times $t_1,t_2$.
\item The stochastic Euler-Poincar\'e equations hold on $\mathfrak{X}^s\times V^*$
\begin{equation}
{\sf d}\frac{\delta \ell}{\delta u} + \mathcal{L}_{{\sf d}\chi_t}\frac{\delta \ell}{\delta u} = \frac{\delta \ell}{\delta a}\diamond a\,dt,
\label{eq:stochep}
\end{equation}
and the advection equation
\begin{equation}
{\sf d}a + \mathcal{L}_{{\sf d}\chi_t}a = 0.
\label{eq:stochadv}
\end{equation}
\end{enumerate}
\end{theorem} 

\begin{proof}
Using integration by parts and the endpoint conditions $v(t_1)=0=v(t_2)$, the variation can be computed to be
\begin{equation}
\begin{aligned}
\delta\int_{t_1}^{t_2}\ell(u,a)\,dt 
&= 
\int_{t_1}^{t_2}\left\langle\frac{\delta\ell}{\delta u},\delta u\right\rangle + \left\langle\frac{\delta\ell}{\delta a},\delta a\right\rangle\,dt\\
&= \int_{t_1}^{t_2}\left\langle\frac{\delta\ell}{\delta u},{\sf d}v + [{\sf d}\chi_t,v]\right\rangle + \left\langle\frac{\delta\ell}{\delta a}\,dt,-\mathcal{L}_v a\right\rangle\\
&= \int_{t_1}^{t_2}\left\langle -{\sf d}\frac{\delta\ell}{\delta u} - \mathcal{L}_{{\sf d}\chi_t}\frac{\delta\ell}{\delta u} + \frac{\delta\ell}{\delta a}\diamond a\,dt,v\right\rangle\\
&= 0\,.
\end{aligned}
\end{equation}
Since the vector field $v$ is arbitrary, one obtains the stochastic Euler-Poincar\'e equations. Finally, the advection equation \eqref{eq:stochadv} follows by applying the KIW formula to $a(t)=g_{t*}a_0$.
\end{proof}

\begin{remark}
The stochastic Euler-Poincar\'e theorem is equivalent to the version presented in \citet{holm2015variational}, which uses stochastic Clebsch constraints. In \cite{holm2015variational} one can also find an investigation the It\^o formulation of the stochastic Euler-Poincar\'e equation. 
\end{remark}

\paragraph{Stochastic Lie-Poisson formulation.}
The stochastic Euler-Poincar\'e equations have an equivalent stochastic Lie-Poisson formulation. To obtain the Lie-Poisson formulation, one must Legendre transform the reduced Lagrangian. The Legendre transformation in the presence of stochasticity becomes itself stochastic in the following way
\begin{equation}
m := \frac{\delta\ell}{\delta u}, \qquad \hslash(m,a)\,dt + \sum_{i=1}^N\langle m,\xi_i\rangle \circ dW_t^i = \langle m,{\sf d}\chi_t\rangle - \ell(u,a)\,dt.
\label{eq:reducedstochlegendre}
\end{equation}
The stochasticity enters the Legendre transformation because the momentum map $m$ is coupled to the stochastic vector field ${\sf d}\chi_t$. The left hand side of the transformation determines the Hamiltonian, which is a semimartingale. The underlying semidirect product group structure has not changed, it is still the $H^s$ diffeomorphisms with a vector space, but the Hamiltonian has become a semimartingale. This implies that in the stochastic case the energy is not conserved, because Hamiltonian depends on time explicitly. Note that \eqref{eq:reducedstochlegendre} emphasises that the Lagrangian does not feature stochasticity in this framework. Instead, the Lagrangian represents the physics in the problem, which does not change. The stochasticity is supposed to account for the difference between observed data and deterministic modelling. The stochastic Lie-Poisson equations are given by
\begin{equation}
{\sf d}(m,a) = -{\rm ad}^*_{(\frac{\delta\hslash}{\delta m},\frac{\delta\hslash}{\delta a})}(m,a)\,dt - \sum_{i=1}^N{\rm ad}^*_{(\xi_i,0)}(m,a)\circ dW_t^i,
\label{eq:stochliepoisson}
\end{equation}
where ${\rm ad}^*$ is given in \eqref{eq:ad*}. Since both the drift and the diffusion part use the same operator in \eqref{eq:stochliepoisson}, the stochastic Lie-Poisson equations preserve the same family of Casimirs (or integral conserved quantities) as the deterministic Lie-Poisson equations. The stochastic Euler-Poincar\'e theorem has a stochastic Kelvin-Noether circulation theorem as a corollary.

Let the manifold $M$ be a submanifold of $\mathbb{R}^n$ with coordinates $X$. Then the volume form can be expressed with respect to a density. That is, $\mu(d^n X) = \rho_0(X)d^n X$. By pushing forward $\rho_0$ along the stochastic flow $g_t$, one obtains $\rho$. Let $\mathfrak{C}^s$ be the space of loops $\gamma:S^1\to\mathfrak{D}^s$, which is acted upon from the left by $\mathfrak{D}^s$. Given an element $m\in\mathfrak{X}^s$, one can obtain a 1-form by formally dividing $m$ by the density $\rho$. 

The circulation map $\mathcal{K}:\mathfrak{C}^s\times V^*\to\mathfrak{X}^{s**}$ is defined by 
\begin{equation}
\langle \mathcal{K}(\gamma,a),m\rangle = \oint_\gamma\frac{m}{\rho}\,.
\end{equation}
Given a Lagrangian $\ell:\mathfrak{X}^s\times V^*\to \mathbb{R}$,  the \emph{Kelvin-Noether quantity} is defined by
\begin{equation}
I(\gamma,u,a) := \oint_\gamma\frac{1}{\rho}\frac{\delta\ell}{\delta u}\,.
\end{equation}
One can now formulate the following stochastic Kelvin-Noether circulation theorem. 
\medskip

\begin{theorem}[Stochastic Kelvin-Noether]\label{Thm:KelThm}
Let $u_t=u(t)$ satisfy the stochastic Euler-Poincar\'e equation \eqref{eq:stochep} and $a_t=a(t)$ the stochastic advection equation \eqref{eq:stochadv}. Let $g_t$ be the flow associated to the vector field ${\sf d}\chi_t$. That is, ${\sf d}\chi_t = {\sf d}g_t\circ g_t^{-1} = u_t\,dt + \sum_{i=1}^N \xi_i\circ dW_t^i$. Let $\gamma_0\in \mathfrak{C}^s$ be a loop. Denote by $\gamma_t = g_t\circ \gamma_0$ and define the Kelvin-Noether quantity $I(t):= I(\gamma_t,u_t,a_t)$. Then
\begin{equation}
{\sf d}I(t) = \oint_{\gamma_t}\frac{1}{\rho}\frac{\delta\ell}{\delta a}\diamond a\,dt\,.
\label{eqn:KelThm}
\end{equation}
\end{theorem}
\begin{proof}
The statement of the stochastic Kelvin-Noether circulation theorem involves a loop that is moving with the stochastic flow. One can transform to stationary coordinates by pulling back the flow to the initial condition. This pull-back yields
\begin{equation}
I(t) = \oint_{\gamma_t}\frac{1}{\rho}\frac{\delta\ell}{\delta u} = \oint_{\gamma_0}g_t^*\left(\frac{1}{\rho}\frac{\delta\ell}{\delta u}\right) = \oint_{\gamma_0}\frac{1}{\rho_0}g_t^*\left(\frac{\delta\ell}{\delta u}\right).
\end{equation}
An application of the Kunita-It\^o-Wentzell formula \eqref{eq:kiwformula} leads to 
\begin{equation}
{\sf d}I(t) = \oint_{\gamma_0}\frac{1}{\rho_0}g_t^*\left({\sf d}\frac{\delta\ell}{\delta u} + \mathcal{L}_{{\sf d}\chi_t}\frac{\delta \ell}{\delta u}\right) = \oint_{\gamma_0}\frac{1}{\rho_0}g_t^*\left(\frac{\delta\ell}{\delta a}\diamond a\right)\,dt,
\end{equation}
since $u$ satisfies the stochastic Euler-Poincar\'e theorem. Transforming back to the moving coordinates by pushing forward with $g_t$ yields the final result.
\end{proof}

Thus, Theorem \ref{Thm:KelThm} explains how particle relabelling symmetry gives rise to the Kelvin-Noether circulation theorem via Noether's theorem. When the only advected quantity present is the mass density, the loop integral of the diamond terms vanishes. This means that circulation is conserved according to Noether's theorem for an incompressible fluid, or for a barotropically compressible fluid. The presence of other advected quantities breaks the symmetry further and introduces the  \emph{diamond terms} which generate circulation, as one can see in the Kelvin-Noether circulation theorem in equation \eqref{eqn:KelThm}. Consequently, the symmetry breaking due to additional order parameters can provide additional mechanisms for the generation of Kelvin-Noether circulation in ideal fluid dynamics. 

\paragraph*{Outlook.}
Stochastic geometric mechanics is an active field of mathematics which has recently established its utility for a broad range of applications in science. Basically, everything that can be done with Hamilton's principle for deterministic geometric mechanics can also be made stochastic in the sense of Stratonovich. This is possible because the variational calculus in Hamilton's principle requires only the product rule and chain rule from ordinary calculus. The happy emergence of the new science of stochastic geometric mechanics was celebrated with the publication of the book \citet{albeverio2017stochastic}. This book showcases some of the recent developments in stochastic geometric mechanics. Another collection of recent developments can be found in \citet{castrillon2019journal}. An ongoing development is in the direction of \emph{rough geometric mechanics}, initiated with a rough version of the Euler-Poincar\'e theorem in \citet{crisan2020variational}. Remarkably, variational principles which are driven by geometric rough paths again only require the product rule and the chain rule. Other directions involve the inclusion of jump processes, fractional derivatives and non-Markovian processes in geometric mechanics. For example, recent work by \citet{albeverio2020weak} shows that SDEs driven by semimartingales with jumps have weak symmetries and a corresponding extension of the reduction and reconstruction technique is discussed.

\section*{Acknowledgments}
We are enormously grateful for many encouraging discussions over the years with T.S. Ratiu, F. Gay-Balmaz, C. Tronci, S. Albeverio, A.B. Cruzeiro, F. Flandoli, and also with our friends in project STUOD (stochastic transport in upper ocean dynamics) and in the geometric mechanics research group at Imperial College London. 
The work of DDH was partially supported by European Research Council (ERC) Synergy
grant STUOD - DLV-856408. EL was supported by EPSRC grant [grant number EP/L016613/1] and is grateful for the warm hospitality at the Imperial College London EPSRC Centre for Doctoral Training in the Mathematics of Planet Earth during the course of this work.

\renewcommand\bibname{\sc References}
\bibliographystyle{plainnat}
\bibliography{biblio}
\end{document}